 \documentclass[letterpaper, 10 pt, conference]{ieeeconf}  

\IEEEoverridecommandlockouts                    

\usepackage[utf8]{inputenc}
\usepackage{lipsum}
\usepackage{graphicx}
\usepackage{amsmath}
\usepackage{amssymb}
\usepackage{indentfirst}
\usepackage[english]{babel}
\usepackage{multirow}
\usepackage{tabularx}
\usepackage{cases}
\usepackage{dirtytalk}
\usepackage{caption}
\usepackage{multicol}
\usepackage[font={small}]{caption}
\usepackage{cite}
\usepackage{makecell}

\usepackage{xcolor}
\providecommand{\keywords}[1]
{
  \small	
  \textbf{\textit{Keywords---}} #1
}

\newtheorem{theorem}{Theorem}
\newtheorem{problem}{Problem}
\newtheorem{lemma}{Lemma}

\newtheorem{definition}{Definition}
\newtheorem{remark}{Remark}
\begin{document}
%
\title{Electric Grid Topology and Admittance Estimation: Quantifying Phasor-based Measurement Requirements}

\author{Norak Rin, Iman Shames, Ian R. Petersen, and Elizabeth L. Ratnam
 \thanks {N. Rin, I. Shames, and I. R. Petersen: School of Engineering, The Australian National University, Australia. (E-mail: norak.rin@anu.edu.au; iman.shames@anu.edu.au; ian.petersen@anu.edu.au). N. Rin is also with CSIRO, Australia. \newline \indent E. L. Ratnam: Department of Electrical and Computer Systems Engineering, Monash University, Australia. (E-mail: liz.ratnam@monash.edu).}}

\maketitle
\begin{abstract}
In this paper, we quantify voltage and current phasor-based measurement requirements for the unique estimation of the electric grid topology and admittance parameters. Our approach is underpinned by the concept of a rigidity matrix that has been extensively studied in graph rigidity theory. Specifically, we show that the rank of the rigidity matrix is the same as that of a voltage coefficient matrix in a corresponding electric power system. Accordingly, we show that there is a minimum number of measurements required to uniquely estimate the admittance matrix and corresponding grid topology. By means of a numerical example on the IEEE 4-node radial network, we demonstrate that our approach is suitable for applications in electric power grids.
\end{abstract}

\IEEEpeerreviewmaketitle

\section{Introduction}

Globally, power systems are transforming in response to climate change, while adapting to withstand more severe and frequent weather events \cite{Comparative, 5618534}. To support power system operators in adapting to the electric grid transformation, enhanced real-time estimation of the grid topology, admittance parameters, and conductor loadings under evolving grid conditions are necessary --- including at the distribution level which has historically had very limited monitoring \cite{1626410}.

As the electric grid transformation progresses, new approaches to real-time grid topology and admittance parameter estimation are emerging \cite{deka2023learning, AReview}. Topology identification of both the transmission and distribution grid coupled with admittance parameter estimation enables the development of control technology to improve grid-resilience \cite{7463503}. Grid-resilience tasks include detecting the location and types of faults occurring in power systems during severe weather events --- which facilitates the safe and swift restoration of the electricity supply \cite{7017458}. 

Many authors have proposed model-based methods to identify the grid topology based on measurement data from grid sensors including phasor measurement units (PMUs) \cite{Estimating, Data, 9094730}. For instance, Deka et al. in \cite{Estimating} propose a graphical model learning algorithm to estimate electrically connected edges in a radial grid topology based on conditional independence tests for voltage measurements. Cavraro et al. in \cite{Data} propose an algorithm to detect switching events that change the topology of the electric grid. Specifically, the data-driven algorithm in \cite{Data} compares voltage measurements from micro-PMUs with a library of signatures representing possible topology changes. 
By using smart meter data, Zhao et al. \cite{9094730} propose a topology identification method of a distribution grid by analysing the nodal voltage correlations based on a Markov Random Field (MRF) structure model. 

By contrast, data-driven methods have been proposed to achieve both topology and admittance parameter estimation \cite{8122027, On, 9858017}. For example, Yu et al. in \cite{8122027} propose a data-driven algorithm to jointly estimate the topology and admittance parameters based on voltage and current measurements and the power flow equations. Specifically, the admittance parameters are estimated via a maximum-likelihood method while the topology is identified by constructing the bus admittance matrix from the estimated admittance parameters. The authors in \cite{On} offer an approach to online topology and admittance estimation based on a least absolute shrinkage and selection operator (LASSO) regression method and a matrix decomposition algorithm that ingress time-series voltage and current phasor measurements. The authors in \cite{9858017} propose an algorithm based on graph theory to estimate the admittance matrix when measurement data is not available at each node.

Several authors have proposed regression-based methods for both topology and admittance parameter estimation. For example, Lang et al. in \cite{Structure} propose an admittance parameter estimation problem using a structured total least squares method with the layout of the problem formulation similar to our study. However, the approach in \cite{Structure} assumes that the grid topology is known in advance. Similarly, Mishra and de Callafon in \cite{9930858} propose an estimation problem to identify three-phase line admittance in the electric grid using a recursive least squares approach with synchrophasor measurements. The approach in \cite{9930858} mentions the required number of measurements to estimate the admittance parameters --- but only for the case where the grid topology is known in advance. None of the aforementioned approaches \cite{Data,Estimating,9094730,8122027, On,9858017,Structure}, with the exception of \cite{9930858}, quantify the number of measurements required to estimate the grid topology and admittance parameters. 

In this paper, we show that there is a minimum number of measurements required to uniquely estimate both the grid topology and admittance parameters for power system networks --- in the absence of any prior knowledge about the grid physics or topology. We show that with voltage and current synchrophasor data available at each node in a transmission or distribution grid, we can estimate the respective grid topology and admittance parameters through the formation of an admittance matrix. Specifically, we investigate the invertibility of a set of equations which are solved to find the admittance parameters. These equations are derived using Kirchhoff's laws. We then develop a connection between this invertibility condition and a rigidity matrix used in graph rigidity theory \cite{Therigidity}. We show that the invertibility condition is equivalent to a condition on the rank of the rigidity matrix. With this rank information, we then determine the minimum number of voltage and current measurements required to estimate the grid topology and admittance parameters. 
 
This paper is organized as follows. In Section 2, we introduce the topology and admittance matrix of an electrical network. In Section 3, we introduce the concept of a rigidity matrix and its relevance to estimating the topology and admittance matrix of the electrical network. In Section 4, we present a simple numerical example for the IEEE 4-node radial network. Section 5 concludes the paper.

\subsection*{Notation}

Let $\mathbb{R}^n$ and $\mathbb{C}^n$ denote the sets of \textit{n}-dimensional vectors of real and complex numbers, respectively. Let $\mathbb{S}^n$ denote the set of $n\times n$ complex symmetric matrices. Throughout, boldface represents vector or matrix quantities whereas scalar quantities are non-bold. Let $\textbf{\textit{X}}^\top$ denote the transpose of a real matrix $\textbf{\textit{X}}$.
Let $\mathbf{1}_{n}\in\mathbb{R}^{n}$ and $\mathbf{0}_{n}\in\mathbb{R}^{n}$ denote an $n$-dimensional vectors which have     all elements equal to $1$ and $0$, respectively. We use the shorthand ( $\dot{}$ ) to denote the derivative operator. Let $\textbf{\textit{D}}=\mathrm{diag}\left(a_1,a_2,\dotsc,a_n\right)$ $\in\mathbb{R}^{n\times n}$ (or $\mathbb{C}^{n\times n}$) denote a diagonal matrix where
\begin{equation}\nonumber
 D_{i,j} = 
    \begin{cases}
      a_i & \text{if $i = j \in\{1,2,\dotsc,n\}$},\\
      0 & \text{if $i \neq j$.}
    \end{cases}       
\end{equation}

\section{Preliminaries}
Our approach to estimate the topology and admittance parameters of an electric grid assumes no prior information, i.e., we assume a complete graph. The topology of an electric grid is rarely characterized by a complete graph, accordingly, we remove edges of the complete graph to determine the grid topology. In what follows, we introduce relevant notation and concepts supporting our approach to estimate the topology and admittance parameters of an electric grid.

Consider a \emph{complete graph} $\mathcal{G}(\mathcal{N,E})$, with the node set $\mathcal{N}= \{0,1,\dotsc,n\}$ and  the edge set $\mathcal{E}= \{(i,j)| i\in\mathcal{N}\; ,j\in\mathcal{N},\; i<j\}$, where edge $(i,j)\in \mathcal{E}$ connects node $i$ to
    node $j$. Note that we use the ordered pair convention for ease of presentation and the underlying graph is assumed to be undirected. The graph $\mathcal{G}(\mathcal{N,E})$ is also assumed to be a
    \emph{weighted graph} in which each edge is weighted by $\alpha_{i,j}$,
    where $(i,j)\in \mathcal{E}$.
Figure~\ref{Fig1}~(a) illustrates an example of the complete and weighted graph $\mathcal{G}(\mathcal{N,E})$, where $n = 3$. 

We define a subgraph $\overline{\mathcal{G}}(\overline{\mathcal{N}},\overline{\mathcal{E}})$ of the graph $\mathcal{G}(\mathcal{N,E})$ by removing the node $0\in \mathcal{N}$ and all associated edges. Here, $\overline{\mathcal{N}} = \mathcal{N}\backslash\{0\} = \{1,2,\dotsc,n\}$ and $\overline{\mathcal{E}}= \{(i,j)| i\in \overline{\mathcal{N}}\; ,j\in\overline{\mathcal{N}},\; i<j\}$. That is, the subgraph $\overline{\mathcal{G}}(\overline{\mathcal{N}},\overline{\mathcal{E}}) \subset \mathcal{G}(\mathcal{N,E})$ is also a complete graph with edges that are weighted. Let $e$ denote the cardinality of the set $\overline{\mathcal{E}}$. Since the subgraph $\overline{\mathcal{G}}(\overline{\mathcal{N}},\overline{\mathcal{E}})$ is a complete graph, $e = n(n-1)/2$. Let $\textbf{\textit{H}} \in \mathbb{R}^{n\times e}$ denote the \emph{incidence matrix} of the undirected subgraph $\overline{\mathcal{G}}(\overline{\mathcal{N}},\overline{\mathcal{E}})$. The $il$-th element of the incidence matrix $\textbf{\textit{H}}$, denoted by $H_{i,l}$ with $i \in \overline{\mathcal{N}}$ and $l \in \overline{\mathcal{E}}$, is defined by 
 \begin{equation*}
     H_{i,l} := 
    \begin{cases}
      1 & \text{if edge $l$ leaves node $i$},\\
      -1 & \text{if edge $l$ enters node $i$},\\
      0 & \text{otherwise}.
    \end{cases}
 \end{equation*}
 \begin{figure}[t]
    \centering
\includegraphics[width=\linewidth]{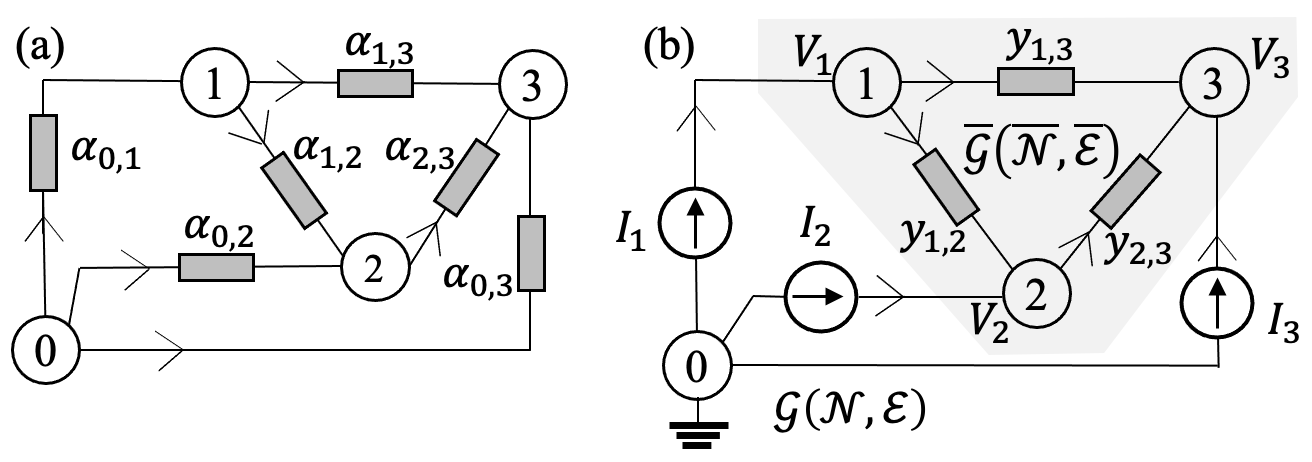}
    \caption{An example graph $\mathcal{G}(\mathcal{N,E})$: (a) where $n=3$; and (b) where $n=3$ and represents a network of admittance. The shaded area represents subgraph $\overline{\mathcal{G}}(\overline{\mathcal{N}},\overline{\mathcal{E}}) \subset \mathcal{G}(\mathcal{N,E})$, where network admittance parameters are represented by $y_{1,3}, y_{1,2}, y_{2,3}$.} \label{Fig1}
\end{figure} 

 Throughout the paper, we use the subgraph $\overline{\mathcal{G}}(\overline{\mathcal{N}},\overline{\mathcal{E}})$ to represent any complete network of admittances. Also, $0\in \mathcal{N}$ represents a $\emph{ground node}$ and the weights along edges $\{(0,1),(0,2),(0,3),\dotsc,(0,n)\} \subset \mathcal{E}$ represent current sources. Each weight on the edges $\overline{\mathcal{E}} \subset \mathcal{E}$ corresponds to an \emph{admittance parameter}, $y_{i,j} \in \mathbb{C}$, $\forall (i,j) \in \overline{\mathcal{E}}$; e.g., the admittance of a transmission line or a transformer. The \emph{voltage} at node $j\in \overline{\mathcal{N}}$ is denoted by $V_j$, where the voltage is measured with respect to ground. Each current source injects current into node $j\in\overline{\mathcal{N}}$, and we denote the injected \emph{current} by $I_j$. This construction is illustrated in Figure~\ref{Fig1}~(b). 

To develop the relationship between the voltages $V_j$ and the currents $I_j$, where $j\in \overline{\mathcal{N}}$, consider the subgraph $\overline{\mathcal{G}}(\overline{\mathcal{N}},\overline{\mathcal{E}})$. For all nodes in subgraph $\overline{\mathcal{G}}(\overline{\mathcal{N}},\overline{\mathcal{E}})$, the \emph{voltage vector} is denoted by $\textbf{\textit{V}} = [V_1,V_2,\dotsc, V_n]^\top \in \mathbb{C}^{n}$ and the \emph{current vector} is denoted by $\textbf{\textit{I}} = [I_1,I_2,\dotsc,I_n]^\top \in \mathbb{C}^{n}$. We denote by $\textbf{\textit{Y}} \in \mathbb{S}^{n}$ an \emph{admittance matrix}, where each element of $\textbf{\textit{Y}}$ is defined by
 \begin{equation*}
     Y_{i,j} := 
    \begin{cases}
      \sum\limits_{\substack{k = 1 \\ k\neq i}}^{n} {y_{i,k}} & \text{if $i=j$},\\
      -y_{i,j} & \text{if $i\neq j$}.
    \end{cases}
 \end{equation*} 
From \cite[Chapter 10]{Basic}, the relationship between $\textbf{\textit{V}}$ and $\textbf{\textit{I}}$, as derived from Kirchhoff's laws, is
 \begin{equation}\label{eq1}
 \textbf{\textit{I}} = \textbf{\textit{Y}}\textbf{\textit{V}},
 \end{equation}
 where $\textbf{\textit{Y}}\mathbf{1}_n = \mathbf{0}_n$. 
 
 Next, we rewrite the relationship between \textbf{\textit{V}} and \textbf{\textit{I}} by incorporating the incidence matrix \textbf{\textit{H}}. Corresponding to the complex symmetric matrix \textbf{\textit{Y}}, we define a complex \emph{admittance vector} by {\textbf{\textit{y}}}:= $\left[
-Y_{2,1},-Y_{3,1},\dotsc,-Y_{n,1},-Y_{3,2},-Y_{4,2},\dotsc,-Y_{n,n-1} \right]^\top \in \mathbb{C}^e$.
 Note that in the definition of \textbf{\textit{y}}, we have used the fact that $\textbf{\textit{Y}}\mathbf{1}_n = \mathbf{0}_n$ to remove the redundant diagonal elements of \textbf{\textit{Y}}; see also \cite{9858017}. With this notation, we can rewrite~\eqref{eq1} as
  \begin{equation}\label{eq2}
      \textbf{\textit{I}} = \textbf{\textit{H}} \mathrm{diag}(\textbf{\textit{H}}^\top \textbf{\textit{V}})\textbf{\textit{y}}.
  \end{equation}
  We denote by $\Tilde{\textbf{\textit{A}}}(\textbf{\textit{V}}) = \textbf{\textit{H}} \mathrm{diag}(\textbf{\textit{H}}^\top \textbf{\textit{V}})\in \mathbb{C}^{n\times e}$ 
  the \emph{voltage coefficient matrix}. Then,~\eqref{eq2} can be equivalently written as
 \begin{equation}\label{eq3}
 \textbf{\textit{I}} = \Tilde{\textbf{\textit{A}}}(\textbf{\textit{V}})\textbf{\textit{y}}.
 \end{equation}

Previously, we considered a single measurement of the voltage vector \textbf{\textit{V}} and the current vector \textbf{\textit{I}}. Next, we consider a series of voltage and current measurements \textbf{\textit{V}} and \textbf{\textit{I}} in order to uniquely identify the grid topology and admittance parameters. Suppose that different measurements of \textbf{\textit{V}} and \textbf{\textit{I}} are indexed by $k$, which is an element in a set $\mathcal{T}=\{1,2,\dotsc,\tau\}$, $k\in\mathcal{T}$. That is, $\textbf{\textit{V}}^{(k)}$ denotes the voltage vector measured at $k\in \mathcal{T}$ and $\textbf{\textit{I}}^{(k)}$ denotes the current vector measured at $k\in \mathcal{T}$. In the sequel, we assume that the pair vectors $(\textbf{\textit{I}}^{(k)}, \textbf{\textit{V}}^{(k)})$ are known exactly for all $k\in\mathcal{T}$ and that the measurements are such that each voltage vector
$\textbf{\textit{V}}^{(1)},\textbf{\textit{V}}^{(2)},\dotsc,\textbf{\textit{V}}^{(\tau)}$ and each current vector
$\textbf{\textit{I}}^{(1)},\textbf{\textit{I}}^{(2)},\dotsc,\textbf{\textit{I}}^{(\tau)}$ are distinct from one another, that is, $V^{(k)}_j \neq V^{(k+1)}_j$ and $I^{(k)}_j \neq I^{(k+1)}_j$, where $j\in\overline{\mathcal{N}}$ and $k\in\mathcal{T}$. It now follows from equation~\eqref{eq3} that 
 \begin{equation}\label{eq4}
     \textbf{\textit{I}}^{(k)} = \Tilde{\textbf{\textit{A}}}(\textbf{\textit{V}}^{(k)})\textbf{\textit{y}}, \forall k \in \mathcal{T},
 \end{equation}
 where $\Tilde{\textbf{\textit{A}}}(\textbf{\textit{V}}^{(k)})$ denotes the voltage coefficient matrix constructed from the voltage vector measured at $k$. Considering all measurements in $\mathcal{T}$, the voltage vectors and the current vectors are stacked as $\textbf{\textit{v}}=\left[\textbf{\textit{V}}^{(1)^{\top}}, \textbf{\textit{V}}^{(2)^{\top}},\dotsc, \textbf{\textit{V}}^{(\tau)^{\top}}\right]^\top \in \mathbb{C}^{n\tau}$ and $\textbf{\textit{i}}=\left[\textbf{\textit{I}}^{(1)^{\top}},\textbf{\textit{I}}^{(2)^{\top}},\dotsc,\textbf{\textit{I}}^{(\tau)^{\top}}\right]^\top \in \mathbb{C}^{n\tau}$. We then rewrite~\eqref{eq4} as
\begin{equation}\label{eq5} \textbf{\textit{i}}=\mathcal{\textbf{A}}(\textbf{\textit{v}})\textbf{\textit{y}},
 \end{equation}
where $\textbf{\textit{A}}(\textbf{\textit{v}}) = \left[
    {\Tilde{\textbf{\textit{A}}}(\textbf{\textit{V}}^{(1)})}^\top,{\Tilde{\textbf{\textit{A}}}(\textbf{\textit{V}}^{(2)})}^\top,\dotsc,{\Tilde{\textbf{\textit{A}}}(\textbf{\textit{V}}^{(\tau)})}^\top\right]^\top \in \mathbb{C}^{n\tau \times e}$. The $\textbf{\textit{A}}(\textbf{\textit{v}})$ matrix is also referred to as the \emph{voltage coefficient matrix} as it also contains $\textbf{\textit{H}}$ and $\textbf{\textit{V}}$, the incidence matrix and the voltage vector, respectively. 

 A \emph{nodal voltage vector} for node $i\in\overline{\mathcal{N}}$ is denoted by $\overline{\textbf{\textit{V}}}_i$ and defined by $\overline{\textbf{\textit{V}}}_i = [V_{i}^{(1)},V_{i}^{(2)},\dotsc, V_{i}^{(\tau)}]^\top \in \mathbb{R}^{\tau}$, where $\{1,2,\dotsc,\tau\} \in \mathcal{T}$. Likewise, A \emph{nodal current vector} for node $i\in\overline{\mathcal{N}}$ is denoted by $\overline{\textbf{\textit{I}}}_i$ and defined by $\overline{\textbf{\textit{I}}}_i = [I_{i}^{(1)},I_{i}^{(2)},\dotsc, I_{i}^{(\tau)}]^\top \in \mathbb{R}^{\tau}$, where $\{1,2,\dotsc,\tau\} \in \mathcal{T}$.

\section{Problem Formulation}\label{Section3}

In this section, we show that there is a minimum number of measurements required to estimate the grid topology and admittance parameters. In an a.c. electric grid, the admittance parameters are complex and in a d.c. electric grid, the admittance parameters are real. Accordingly, in Section~\ref{Section3A}, we consider the simpler case where $\textbf{\textit{Y}}$ is real; i.e., $\textbf{\textit{Y}}\in \mathbb{R}^{n\times n}$, where $\textbf{\textit{Y}}$ represents a conductance matrix. That is, $\textbf{\textit{y}} \in \mathbb{R}^{e}$ represents an unknown conductance vector. Then in Section~\ref{Section3B}, we consider the general a.c. case where $\textbf{\textit{Y}}$ is a complex matrix. 

\subsection{Conductance matrix estimation}\label{Section3A}

We first consider the simpler d.c. case where $\textbf{\textit{Y}}$ is real. 

\begin{problem}\label{p1}
   Consider a conductance network as defined by
   $\overline{\mathcal{G}}(\overline{\mathcal{N}},\overline{\mathcal{E}})$ with an unknown conductance vector $\textbf{\textit{y}} \in \mathbb{R}^{e}$.
   Suppose that $\tau$ measurements are carried out where for each measurement
   index $k\in \mathcal{T}$ the d.c. current applied to the network is
   $\textbf{\textit{I}}^{(k)} \in \mathbb{R}^{n}$ and the corresponding voltages
   $\textbf{\textit{V}}^{(k)} \in \mathbb{R}^{n}$ are measured. 
   We seek to find the minimum number of measurements
   $\tau$ required to uniquely determine the unknown conductance vector
   $\textbf{\textit{y}}$. \end{problem}

 Solving Problem~\ref{p1} is equivalent to finding conditions on $\textbf{\textit{A}}(\textbf{\textit{v}})$ that produce a unique solution to equation~\eqref{eq5}. From the Rouché-Capelli Theorem in \cite[Chapter 2]{Linear}, a unique solution exists if and only if the rank of the voltage coefficient matrix $\textbf{\textit{A}}(\textbf{\textit{v}})$ is equal to the number of unknown conductances $e$. Since $\overline{\mathcal{G}}(\overline{\mathcal{N}},\overline{\mathcal{E}})$ is a complete graph, the number of unknown conductances $e$ is equal to $n(n-1)/2$. Thus, we seek to find the minimum value of $\tau$ for which 
 \begin{equation}\label{eq6}
     \mathrm{rank}(\textbf{\textit{A}}(\textbf{\textit{v}})) = \frac{n(n-1)}{2}.
 \end{equation}

In what follows we introduce the concept of rigid frameworks, globally rigid frameworks, and relating definitions in order to establish the rank of $\textbf{\textit{A}}(\textbf{\textit{v}})$; see \cite{Therigidity, Conditions, Globally, gortler2014generic, jordan2017global, egres-14-12} for more details. 

\begin{definition}[Graph Realization {\cite[Section 1]{Conditions}}]\label{d1}
     A \emph{realization} of a graph \emph{$\overline{\mathcal{G}}$} is a mapping \emph{$\textbf{\textit{x}}$} from the nodes of \emph{$\overline{\mathcal{G}}$} to points in Euclidean space.
\end{definition}
\begin{definition}[Framework {\cite[Section 1.2]{egres-14-12}}]\label{d2}
     A \emph{$\tau$}-dimensional \emph{framework} is a pair \emph{$(\overline{\mathcal{G}}, \textbf{\textit{x}})$}, where \emph{$\overline{\mathcal{G}}$} is a graph and \emph{$\textbf{\textit{x}}$} is a realization that maps the nodes of \emph{$\overline{\mathcal{G}}$} to points in the $\tau$-dimensional Euclidean space $\mathbb{R}^\tau$.
\end{definition}
\begin{definition}[Generic Frameworks {\cite[Defintion 6]{gortler2014generic}}]\label{d3}
     A framework \emph{$(\overline{\mathcal{G}}, \textbf{\textit{x}})$} is \emph{generic} if the realization of the nodes do not satisfy any non-zero algebraic equation with rational coefficients. 
\end{definition} 
\begin{remark}[A nongeneric framework example]\label{r1}
    For the framework $(\overline{\mathcal{G}},\textbf{\textit{x}})$, $\textbf{\textit{x}}_i \in \mathbb{R}^{\tau}$ represents a point (mapped from the nodes of $\overline{\mathcal{G}}$) in $\mathbb{R}^\tau$, where $i\in\overline{\mathcal{N}}$. We write $\textbf{\textit{x}}_i = \left[x_i^{(1)},x_i^{(2)}, \dotsc,x_i^{(\tau)}\right]^\top \in \mathbb{R}^{\tau}$, and define $\textbf{\textit{x}} = \left[\textbf{\textit{x}}_1^\top,\textbf{\textit{x}}_2^\top, \dotsc, \textbf{\textit{x}}_n^\top \right]^\top \in \mathbb{R}^{n \tau}.$
 If the points $[\textbf{\textit{x}}_1,\textbf{\textit{x}}_2, \dotsc, \textbf{\textit{x}}_n]$ are linearly dependent, then the framework $(\overline{\mathcal{G}},\textbf{\textit{x}})$ is not generic.
\end{remark}

\begin{definition}[Finite Flexing {\cite[Section 2.1]{Conditions}}] \label{d4}
      A \emph{finite flexing} of a framework $(\overline{\mathcal{G}}, \textbf{\textit{x}})$ is a family of realizations  of $\overline{\mathcal{G}}$, parameterized by time $t$ so that the location of each node $i\in\overline{\mathcal{N}}$ is a differentiable function of $t$, and \begin{equation}\label{eq7}
    ||\textbf{\textit{x}}_i(t) - \textbf{\textit{x}}_j(t)||^2 = c_{i,j}, \forall (i,j) \in \overline{\mathcal{E}},
\end{equation} where each $c_{i,j}$ is a constant, and $||\textbf{\textit{x}}_i(t) - \textbf{\textit{x}}_j(t)||^2$ is the Euclidean distance between $\textbf{\textit{x}}_i$ and $\textbf{\textit{x}}_j$.
\end{definition}

\begin{remark}[Infinitesimal Motion {\cite[Section 2.1]{Conditions}}]\label{r2} 
  Differentiating both sides of~\eqref{eq7} with respect to $t$ yields
\begin{equation}\label{eq8}
    (\textbf{\textit{x}}_i - \textbf{\textit{x}}_j)^\top (\dot{\textbf{\textit{x}}}_{i} - \dot{\textbf{\textit{x}}}_{j}) = 0, \forall (i,j) \in \overline{\mathcal{E}},
\end{equation} where $\dot{\textbf{\textit{x}}}_{i}$ is the velocity of node $i$. Equation~\eqref{eq8} describes an \emph{infinitesimal motion} of the framework $(\overline{\mathcal{G}}, \textbf{\textit{x}})$.
\end{remark}

\begin{definition}[Trivial Infinitesimal Motion {\cite[Section 2.1]{Conditions}}]\label{d5}
     A \emph{trivial infinitesimal motion} refers
    to a rotation or translation of a framework. This satisfies the definition
    of a finite flexing. 
    \end{definition}

\begin{remark}[Number of trivial infinitesimal motions]\label{r3}
From \cite[Section 2.1]{Conditions} we know that, for a framework in $\tau$-dimension, there are $\tau$
    independent translations and $\tau(\tau-1)/2$ rotations. 
\end{remark}

\begin{definition}[Infinitesimal Rigidity {\cite[Section 2.1]{Conditions}}]\label{d6}
      A framework \emph{$(\overline{\mathcal{G}}, \textbf{\textit{x}})$} is \emph{infinitesimally rigid} if it has only trivial infinitesimal motions.
\end{definition}

\begin{remark}[Infinitesimal rigidity is (generic) rigidity] \label{r4}
    From \cite[Theorem 2.1]{Conditions} we know that, if a framework \emph{$(\overline{\mathcal{G}}, \textbf{\textit{x}})$} is infinitesimally rigid, it is also (generically) rigid.
\end{remark}
    

\begin{definition}[Equivalent Frameworks {\cite[Section 1.2]{egres-14-12}}]\label{d7}
     Two $\tau$-dimensional frameworks \emph{$(\overline{\mathcal{G}}, \textbf{\textit{p}})$} and \emph{$(\overline{\mathcal{G}}, \textbf{\textit{q}})$} are \emph{equivalent} if $||\textbf{\textit{p}}_i - \textbf{\textit{p}}_j||^2 = ||\textbf{\textit{q}}_i - \textbf{\textit{q}}_j||^2$ for all $(i,j) \in \overline{\mathcal{E}}$, that is, the Euclidean distance between node point $i$ and $j$ in the two frameworks \emph{$(\overline{\mathcal{G}}, \textbf{\textit{p}})$} and \emph{$(\overline{\mathcal{G}}, \textbf{\textit{q}})$} are the same for all node points connected by an edge.
\end{definition}
\begin{definition}[Congruent Frameworks {\cite[Section 1.2]{egres-14-12}}]\label{d8}
     Two $\tau$-dimensional frameworks \emph{$(\overline{\mathcal{G}}, \textbf{\textit{p}})$} and \emph{$(\overline{\mathcal{G}}, \textbf{\textit{q}})$} are \emph{congruent} if $||\textbf{\textit{p}}_i - \textbf{\textit{p}}_j||^2 = ||\textbf{\textit{q}}_i - \textbf{\textit{q}}_j||^2$ for all $i,j \in \overline{\mathcal{N}}$, that is, the Euclidean distance between node point $i$ and $j$ in the two frameworks \emph{$(\overline{\mathcal{G}}, \textbf{\textit{p}})$} and \emph{$(\overline{\mathcal{G}}, \textbf{\textit{q}})$} are the same for all node points.
\end{definition}

\begin{definition}[Globally Rigid Frameworks {\cite[Section 1.2]{egres-14-12}}]\label{d9}
    A $\tau$-dimensional framework \emph{$(\overline{\mathcal{G}}, \textbf{\textit{x}})$} is \emph{globally rigid} if any other framework that is equivalent to \emph{$(\overline{\mathcal{G}}, \textbf{\textit{x}})$} is also congruent to \emph{$(\overline{\mathcal{G}}, \textbf{\textit{x}})$}.
\end{definition}
\begin{definition}[Globally Rigid Graphs {\cite[Section 1.3]{egres-14-12}}]\label{d10}
     A graph $\overline{\mathcal{G}}$ is \emph{globally rigid} in $\mathbb{R}^{\tau}$ if every generic realization of $\overline{\mathcal{G}}$ in $\mathbb{R}^\tau$ is globally rigid.
\end{definition}
\begin{figure}[h!]
    \centering
\includegraphics[width=\linewidth]{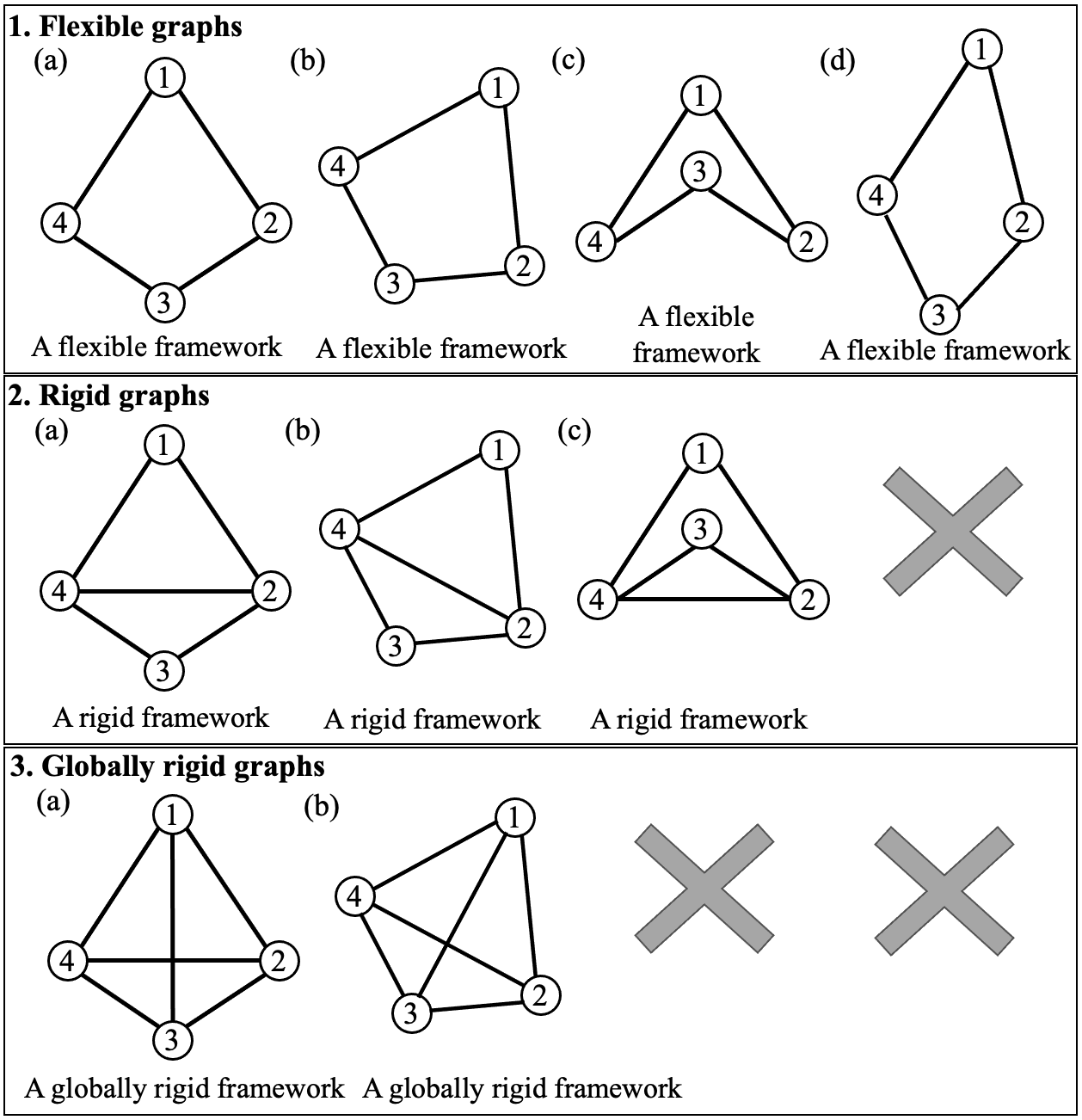}
    \caption{Graphs in two dimensions: 1. flexible; 2. rigid; and 3. globally rigid. That is, 1(a)-(d) are flexible frameworks, 2(a)-(c) are rigid frameworks, and 3(a)-(b) are globally rigid frameworks. } \label{Fig3}
\end{figure}

In Figure~\ref{Fig3}, we present an example of a flexible, a rigid, and a globally rigid graph in a two-dimensional plane. In Figure~\ref{Fig3}.1, the graph is flexible since all of its frameworks (a)-(d) are flexible. Figure~\ref{Fig3}.2 indicates a rigid graph, although it is not globally rigid. This observation is due to framework (c) is not congruent to framework (a). However, individually, each framework (a), (b), and (c) in Figure~\ref{Fig3}.2 are rigid. In Figure~\ref{Fig3}.3, we observe a globally rigid graph since framework (b) is both equivalent and congruent to framework (a). In fact, framework (b) is a rotation of framework (a). Moreover, in Figure~\ref{Fig3} we observe each framework for each respective graph is generic in the two-dimension plane.

\begin{remark}[Complete graphs are globally rigid]\label{r5}
    Following from the example presented in Figure~\ref{Fig3} and \cite[Theorem 4.1.2]{egres-14-12}, a graph $\overline{\mathcal{G}}$ is globally rigid if and only if $\overline{\mathcal{G}}$ is a complete graph. This statement, along with Definition \ref{d10}, implies that, if $\overline{\mathcal{G}}$ is a complete graph in $\mathbb{R}^\tau$, then $\overline{\mathcal{G}}$ is generically globally rigid in $\mathbb{R}^\tau$.
\end{remark}

 Next, we draw a connection between the concept of rigid frameworks and the network of admittances represented by $\overline{\mathcal{G}}(\overline{\mathcal{N}},\overline{\mathcal{E}})$.  Since $\overline{\mathcal{G}}$ is a complete graph, the corresponding framework $(\overline{\mathcal{G}}, \textbf{\textit{x}})$ is rigid. Thus, we can use~\eqref{eq8} to describe the rigidity of $(\overline{\mathcal{G}}, \textbf{\textit{x}})$. Equivalently, we can write~\eqref{eq8} as
\begin{equation}\label{eq9}
    \sum_{k = 1}^{\tau}{(x_{i}^{(k)} - x_{j}^{(k)})(\dot{x}_{i}^{(k)}-\dot{x}_{j}^{(k)})} = 0, \forall (i,j) \in \overline{\mathcal{E}}.
\end{equation} 
Also, we can write~\eqref{eq9} in matrix form as
\begin{equation}\label{eq10}
    \textbf{\textit{R}}(\textbf{\textit{x}}) \dot{\textbf{\textit{x}}} = \mathbf{0}_e,
\end{equation}
where the matrix $\textbf{\textit{R}}(\textbf{\textit{x}})\in \mathbb{R}^{e\times n\tau}$ is referred to as the \emph{rigidity matrix} of the framework, see \cite{Conditions}, and $\dot{\textbf{\textit{x}}} = \left[
\dot{\textbf{\textit{x}}}_{1}^\top,,\dot{\textbf{\textit{x}}}_{2}^\top,\dotsc,\dot{\textbf{\textit{x}}}_{n}^\top\right]^\top \in \mathbb{R}^{n\tau}$, where $\dot{\textbf{\textit{x}}}_{i} = \left[\dot{x}_{i}^{(1)},\dot{x}_{i}^{(2)},\dotsc,\dot{x}_{i}^{(\tau)}\right]^\top \in \mathbb{R}^{\tau}$.
Thus, $\dot{\textbf{\textit{x}}}$ is a null vector of $\textbf{\textit{R}}(\textbf{\textit{x}})$. The following lemma shows the relationship between the matrices $\textbf{\textit{R}}(\textbf{\textit{x}})$ and ${\textbf{\textit{A}}(\textbf{\textit{v}})}^\top$ when we choose $\textbf{\textit{x}}= \textbf{\textit{v}}$.

\begin{lemma}\label{l1}
 The matrix ${\textbf{\textit{A}}(\textbf{\textit{v}})}^{\top} \in \mathbb{R}^{e \times n\tau}$ has the same null space as the rigidity matrix $\textbf{\textit{R}}(\textbf{\textit{x}})\in \mathbb{R}^{e \times n\tau}$ when $\textbf{\textit{v}} = \textbf{\textit{x}}$.
\end{lemma}

Lemma~\ref{l1} implies that both ${\textbf{\textit{A}}(\textbf{\textit{v}})}^\top$ and $\textbf{\textit{R}}(\textbf{\textit{x}})$ have the same null space and the same column space dimension $n\tau$, when $\textbf{\textit{v}} = \textbf{\textit{x}}$. Accordingly, the rank of ${\textbf{\textit{A}}(\textbf{\textit{v}})}^\top$ and $\textbf{\textit{R}}(\textbf{\textit{x}})$ are also the same when $\textbf{\textit{v}} = \textbf{\textit{x}}$. Therefore, to investigate the rank of ${\textbf{\textit{A}}(\textbf{\textit{v}})}^\top$ we can equivalently study the rank of $\textbf{\textit{R}}(\textbf{\textit{x}})$. The rank of $\textbf{\textit{R}}(\textbf{\textit{x}})$ is obtained in Lemma~\ref{l2} below.

\begin{lemma}\label{l2}
     For a framework in $\mathbb{R}^\tau$ which corresponds to a complete graph, assume that the vectors $[\textbf{\textit{x}}_1,\textbf{\textit{x}}_2, \dotsc, \textbf{\textit{x}}_n]$ are generic in the sense of Definition~\ref{d3}. Then the rank of the corresponding rigidity matrix $\textbf{\textit{R}}(\textbf{\textit{x}})$ is equal to $n\tau - \tau(\tau + 1)/2$ when $n\geq \tau$. 
\end{lemma}


In the following theorem, we derive a formula for the minimum number of measurements $\tau$ required to estimate the grid topology and conductance parameters. 

\begin{theorem}\label{t1}
Assuming that the nodal voltage vectors $[\overline{\textbf{\textit{V}}}_1, \overline{\textbf{\textit{V}}}_2,\dotsc, \overline{\textbf{\textit{V}}}_n]$ are generic in the sense of Definition~\ref{d3}, the minimum number of measurements $\tau$ such that $\textbf{\textit{i}}=\mathcal{\textbf{A}}(\textbf{\textit{v}})\textbf{\textit{y}}$ (see equation~\eqref{eq5}) has a unique solution $\textbf{\textit{y}}$ is $\tau = n-1$.
\end{theorem}

\begin{remark}\label{r6}
If nodal voltage vectors $[\overline{\textbf{\textit{V}}}_1, \overline{\textbf{\textit{V}}}_2,\dotsc, \overline{\textbf{\textit{V}}}_n]$
are linearly dependent, then as discussed in Remark~\ref{r1}, they are not generic as per Definition~\ref{d3}. 
\end{remark}
\subsection{Admittance matrix estimation}\label{Section3B}

We now consider the general a.c. case where $\textbf{\textit{Y}}\in \mathbb{C}^{n\times n}$ is a complex matrix. That is, the admittance matrix encapsulates both conductance and susceptance. We assume that a series of sinusoidal current $\textbf{\textit{i}} \in \mathbb{C}^{n\tau}$ are applied and $\tau$ measurements are collected. For applications in power systems, we assume the sinusoidal currents are applied at the same frequency and $\tau$ synchrophasor measurements (both the current and voltage phasors) are collected at each node via precise PMUs.

\begin{problem}\label{p2}
   Consider an admittance network as defined by
   $\overline{\mathcal{G}}(\overline{\mathcal{N}},\overline{\mathcal{E}})$ with
   an unknown admittance vector $\textbf{\textit{y}}\in \mathbb{C}^{e}$. Suppose
   that $\tau$ measurements are carried out where for each measurement index
   $k\in \mathcal{T}$ the current applied to the network is
   $\textbf{\textit{I}}^{(k)} \in \mathbb{C}^{n}$ and the corresponding voltage
   $\textbf{\textit{V}}^{(k)} \in \mathbb{C}^{n}$ is measured. Assuming that the
   nodal voltage vectors $[\overline{\textbf{\textit{V}}}_1,
   \overline{\textbf{\textit{V}}}_2,\dotsc, \overline{\textbf{\textit{V}}}_n]$
   are generic, we seek to find the minimum number of measurements $\tau$
   required to uniquely determine the unknown admittance vector
   $\textbf{\textit{y}}$. 
   \end{problem}


We can define the concept of rigid frameworks in the complex space the same way as in the real space. Now, a realization of a graph $\overline{\mathcal{G}}$ is a mapping $\textbf{\textit{x}}$ from the nodes of \emph{$\overline{\mathcal{G}}$} to points in the complex space. Accordingly, the voltage coefficient matrix and the rigidity matrix are in the complex space, that is, ${\textbf{\textit{A}}(\textbf{\textit{v}})}^{\top} \in \mathbb{C}^{e \times n\tau}$ and $\textbf{\textit{R}}(\textbf{\textit{x}})\in \mathbb{C}^{e \times n\tau}$. Similar to the real case, we note that ${\textbf{\textit{A}}(\textbf{\textit{v}})}^{\top}$ has the same null space as $\textbf{\textit{R}}(\textbf{\textit{x}})$ when $\textbf{\textit{v}}= \textbf{\textit{x}}$. Hence, ${\textbf{\textit{A}}(\textbf{\textit{v}})}^{\top}$ has the same rank and nullity as $\textbf{\textit{R}}(\textbf{\textit{x}})$ when $\textbf{\textit{v}} = \textbf{\textit{x}}$.

To find the rank of the rigidity matrix $\textbf{\textit{R}}(\textbf{\textit{x}})$ in the complex space $\mathbb{C}$, we also briefly introduce the concept of a rigidity matriod; see \cite{Globally, oxley2011matroid} for more details. The \emph{rigidity matroid} of a graph $\overline{\mathcal{G}}$, denoted by $\mathcal{R}(\overline{\mathcal{G}})$, is a matroid defined on the edges set of $\overline{\mathcal{G}}$ which reflects the rigidity properties of all generic realizations of $\overline{\mathcal{G}}$ \cite[Section 2.2]{Globally}. According to \cite[Section 6.1]{oxley2011matroid}, a $b$-element rank-$r$ matroid $\mathcal{M}$ is representable in a field $\mathbb{F}$ if and only if $\mathcal{M}$ is isomorphic to $\mathcal{M}(\textbf{\textit{X}})$ for some rank-$r$ $a\times b$ matrix $\textbf{\textit{X}}$ in $\mathbb{F}$ with $a\geq r$. As for the rigidity matroid $\mathcal{R}(\overline{\mathcal{G}})$, a linear representation of the rigidity matroid $\mathcal{R}(\overline{\mathcal{G}})$ in $\mathbb{R}$ or $\mathbb{C}$ is the rigidity matrix $\textbf{\textit{R}}(\overline{\mathcal{G}},\textbf{\textit{x}})$ in $\mathbb{R}$ or $\mathbb{C}$, respectively.

Next, we provide the rank of $\textbf{\textit{A}}(\textbf{\textit{v}})$ in the complex case as in Lemma \ref{l3} below.

\begin{lemma}\label{l3}
     For the general a.c. case where $\textbf{\textit{A}}(\textbf{\textit{v}}) \in \mathbb{C}^{n\tau \times e}$, the rank of $\textbf{\textit{A}}(\textbf{\textit{v}})$ is equal to $n\tau - \tau(\tau + 1)/2$ when $n\geq \tau$.  
\end{lemma}

 The following theorem provides a solution to Problem~\ref{p2}.
\begin{theorem}\label{t2}
Consider the general a.c. case where $\textbf{\textit{Y}}\in \mathbb{C}^{n\times n}$ is a complex matrix. Assuming that the nodal voltage vectors $[\overline{\textbf{\textit{V}}}_1, \overline{\textbf{\textit{V}}}_2,\dotsc, \overline{\textbf{\textit{V}}}_n]$ are generic in the sense of Definition~\ref{d3}, the minimum number of measurements $\tau$ such that $\textbf{\textit{i}}=\mathcal{\textbf{A}}(\textbf{\textit{v}})\textbf{\textit{y}}$ (see~\eqref{eq5}) has a unique solution $\textbf{\textit{y}}$ is $\tau = n-1$.
\end{theorem}

\begin{remark}\label{r7} All proofs are included in the arXiv paper \cite{rin2024electricgridtopologyadmittance}. In \cite{rin2024electricgridtopologyadmittance}, we took a longer approach to prove Theorem~\ref{t1} compared to the proof for Theorem~\ref{t2} to show the connection between the concept of graph rigidity and the network of admittances in a power system as this relationship has not been previously known. 
\end{remark}

\begin{remark}\label{r8}
Here, we outline our approach to estimate the electric grid topology.  Specifically, starting with the complete graph $\overline{\mathcal{G}}$, collect $\tau$ measurements to find a unique solution $\textbf{\textit{y}}$. Following Theorem~\ref{t2}, the unique solution $\textbf{\textit{y}}$ corresponds to $\tau = n-1$. Then, find admittance parameters $Y_{i,j}=0$ along each edge $(i,j) \in \overline{\mathcal{E}}$, and remove these edges. The resulting graph with non-zero admittance parameters is the electric grid topology.
\end{remark}

Our approach to electric grid topology and admittance estimation can be applied to any kind of electrical network (e.g., radial, meshed, three-phase, single-phase, etc). Our approach considers a complete graph as it reflects the case of not having prior information of the yet-to-be-estimated grid topology. Then, our estimation process uncovers edges to be removed from the complete graph. Accordingly, the estimated grid topology is not a complete graph.

\section{Numerical example}

In Figure~\ref{Fig4}~(a), we present the admittance network for subgraph $\overline{\mathcal{G}}(\overline{\mathcal{N}},\overline{\mathcal{E}})$, with $n = 4$. In this section, we illustrate the result of Theorem~\ref{t2} for the general a.c. case where the admittance network is defined by subgraph $\overline{\mathcal{G}}(\overline{\mathcal{N}},\overline{\mathcal{E}})$ as presented in Figure~\ref{Fig4}~(a). Specifically, we illustrate the minimum number $\tau$ of voltage measurements and corresponding current measurements required to estimate the topology and admittance parameters of the network. 

\begin{figure}[!ht]
    \centering
\includegraphics[width=0.84\linewidth]{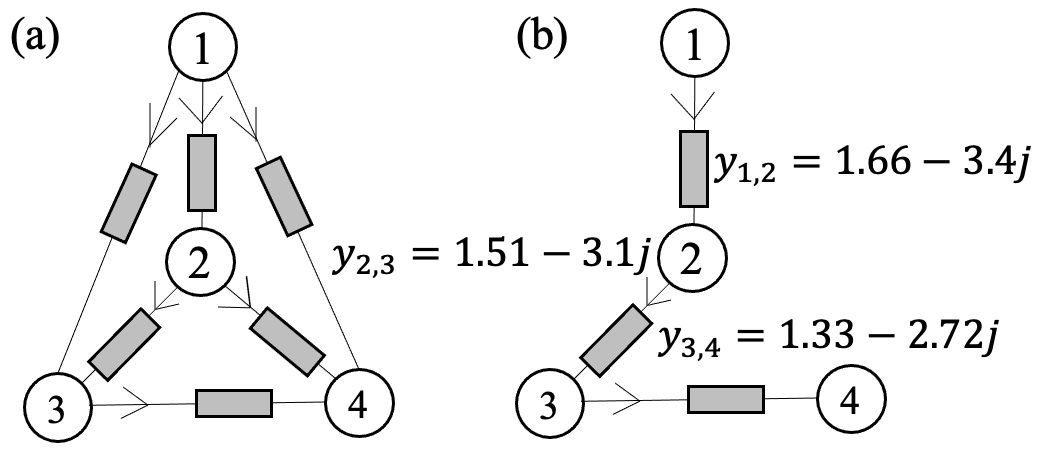}
    \caption{(a) Admittance network as defined by subgraph $\overline{\mathcal{G}}(\overline{\mathcal{N}},\overline{\mathcal{E}})$, where $n = 4$ and; (b) the corresponding estimated topology and admittance parameters of the network.} \label{Fig4}
\end{figure}

\begin{table}[h]
    \centering
    \captionof{table}{Measurements for subgraph $\overline{\mathcal{G}}(\overline{\mathcal{N}},\overline{\mathcal{E}})$, where $n = 4$.}\label{Tab1}
    \begin{tabular}{ *{3}{c} }
\hline
Node ($k = 1$)  & $\textbf{\textit{V}}^{(1)}$ & $\textbf{\textit{I}}^{(1)}$ \\
\hline
$1$  & $V_1^{(1)} = 12470 + 0.1679\jmath$ & $I_1^{(1)} = 8777 - 18290\jmath$  \\
\hline
$2$  & $V_2^{(1)} = 7107 - 36.77\jmath$ & $I_2^{(1)} = -1770 +3098\jmath$\\
\hline
$3$ & $V_3^{(1)} =2243 - 144.4\jmath$ & $I_3^{(1)} =-6975 + 14032\jmath$  \\
\hline
$4$ & $V_4^{(1)} =1894 - 303.0\jmath$ & $I_4^{(1)} =-32.08 + 1160\jmath$ \\
\hline
\end{tabular}
\begin{tabular}{ *{3}{c} }
Node ($k = 2$)  & $\textbf{\textit{V}}^{(2)}$ & $\textbf{\textit{I}}^{(2)}$ \\
\hline
$1$  & $V_1^{(2)} =12502 + 13.00\jmath$ & $I_1^{(2)} =8578 - 18349\jmath$  \\
\hline
$2$  & $V_2^{(2)} =7148 - 78.00\jmath$ & $I_2^{(2)} =-1394 +2941\jmath$\\
\hline
$3$ & $V_3^{(2)} =2205 - 167.0\jmath$ & $I_3^{(2)} =-7183 + 14503\jmath$  \\
\hline
$4$ & $V_4^{(2)} =1936 - 298.0\jmath$ & $I_4^{(2)} =-0.9120 + 905.6\jmath$ \\
\hline
\end{tabular}
\begin{tabular}{ *{3}{c} }
Node ($k = 3$)  & $\textbf{\textit{V}}^{(3)}$ & $\textbf{\textit{I}}^{(3)}$ \\
\hline
$1$  & $V_1^{(3)} =12548 + 59.00\jmath$ & $I_1^{(3)} =8417 - 18424\jmath$  \\
\hline
$2$  & $V_2^{(3)} =7195 - 79.00\jmath$ & $I_2^{(3)} =-1013 +2809\jmath$\\
\hline
$3$ & $V_3^{(3)} =2170 - 137.0\jmath$ & $I_3^{(3)} =-7692 + 14847\jmath$  \\
\hline
$4$ & $V_4^{(3)} =1984 - 334.0\jmath$ & $I_4^{(3)} =288.8 + 767.5\jmath$ \\
\hline
\end{tabular}
\vspace{-10pt}
\end{table}

Before estimating the grid topology and admittance parameters corresponding to subgraph $\overline{\mathcal{G}}(\overline{\mathcal{N}},\overline{\mathcal{E}})$, we must first collect voltage and current measurements for each node as illustrated in Figure~\ref{Fig4}~(a). In what follows, we consider three measurements of the voltage $[\textbf{\textit{V}}^{(1)},\textbf{\textit{V}}^{(2)},\textbf{\textit{V}}^{(3)}]$ and the corresponding current $[\textbf{\textit{I}}^{(1)},\textbf{\textit{I}}^{(2)},\textbf{\textit{I}}^{(3)}]$. Table~\ref{Tab1} presents the voltage and current measurements at $k = 1, 2, 3$, where $k\in\mathcal{T}$, for the subgraph $\overline{\mathcal{G}}(\overline{\mathcal{N}},\overline{\mathcal{E}})$.

We took the following approach to collect voltage and corresponding current measurements for subgraph $\overline{\mathcal{G}}(\overline{\mathcal{N}},\overline{\mathcal{E}})$ as presented in Table~\ref{Tab1}. First, we considered the IEEE 4-node feeder \cite{IEEE}. Specifically, we modify the IEEE 4-node feeder by replacing the transformer between nodes $2$ and $3$ with an admittance of $y_{2,3} = 1.51 - 3.1\jmath$. We then construct the injected current $I_j$ at each node $j\in\overline{\mathcal{N}}$, adhering to Kirchhoff's current and voltage laws, where the nodal voltage $V_j$ at each node $j\in\overline{\mathcal{N}}$ is as specified by the IEEE 4-node feeder \cite{IEEE}. The voltage and current measurements for each of the four nodes in the modified IEEE 4-node feeder were then applied to the respective nodes in subgraph $\overline{\mathcal{G}}(\overline{\mathcal{N}},\overline{\mathcal{E}})$.  

\begin{table}[h]
    \centering
    \captionof{table}{Required number of measurements $\tau$.}\label{Tab2}
    \begin{tabular}{ *{4}{c} }
\hline
\makecell{Number of \\ measurements ($\tau$)}  & $\mathrm{rank}(\textbf{\textit{A}}(\textbf{\textit{v}}))$ & \makecell{Number of \\ unknowns} & Solution \\
\hline
$\tau = 1$  & 3 & 6 & Not unique  \\
\hline $\tau = 2$  & 5 & 6 & Not unique \\
\hline $\tau = 3$ & 6 & 6 & Unique  \\
\hline
\end{tabular}
\vspace{-3pt}
\end{table}

Next, we obtain the admittance vector $\textbf{\textit{y}}$ in~\eqref{eq5} by multiplying $\textbf{\textit{i}}$ by the pseudo-inverse of $\textbf{\textit{A}}(\textbf{\textit{v}})$.  Specifically, the elements of the estimated admittance vector $\textbf{\textit{y}}$ are $y_{1,2} = 1.66 - 3.4\jmath$, $y_{2,3} = 1.51 - 3.1\jmath$, and $y_{3,4} = 1.33 - 2.72\jmath$, and all other admittance parameters are zero. Figure~\ref{Fig4}~(b) presents the corresponding estimated topology and admittance parameters. In Figure~\ref{Fig4}~(b), we observe that the estimated admittance vector $\textbf{\textit{y}}\in \mathbb{C}^{e}$ corresponds to the admittance vector of the modified IEEE 4-node feeder \cite{IEEE}. Note that the estimated grid topology is not a complete graph (see Remark~\ref{r8}).

Table~\ref{Tab2} presents the number of measurements $\tau$ required to determine a unique solution to the system of equations~\eqref{eq5} for the network as illustrated in Figure~\ref{Fig4} (a). When $\tau = 2$, 
the solution is not unique. However, when $\tau = 3$, $\mathrm{rank}(\textbf{\textit{A}}(\textbf{\textit{v}}))$ is equal to the number of unknowns, and thus the solution is unique. This is in agreement with Theorem~\ref{t2}, where at least $3$ measurements ($\tau = n-1$) are required.
\section{Conclusion}

In this paper, we considered the minimum number of voltage and current measurements required to uniquely estimate the topology and admittance parameters of an electric power grid. We showed that a rank condition for the voltage coefficient matrix of an admittance network was equivalent to a corresponding rank condition for a rigidity matrix arising in graph rigidity theory. The admittance matrix and corresponding topology were uniquely estimated when the number of voltage and current measurements was one less than the number of nodes in the admittance network. By means of a numerical example based on the IEEE 4-node radial network, we validated our proposed estimation method. 

\section*{Acknowledgement}

The first author would like to thank Dr. Nariman Mahdavi Mazdeh from CSIRO for his helpful discussions.
\bibliographystyle{IEEEtran}
\bibliography{ref.bib}

\begin{thebibliography}{10}
\providecommand{\url}[1]{#1}
\csname url@samestyle\endcsname
\providecommand{\newblock}{\relax}
\providecommand{\bibinfo}[2]{#2}
\providecommand{\BIBentrySTDinterwordspacing}{\spaceskip=0pt\relax}
\providecommand{\BIBentryALTinterwordstretchfactor}{4}
\providecommand{\BIBentryALTinterwordspacing}{\spaceskip=\fontdimen2\font plus
\BIBentryALTinterwordstretchfactor\fontdimen3\font minus \fontdimen4\font\relax}
\providecommand{\BIBforeignlanguage}[2]{{%
\expandafter\ifx\csname l@#1\endcsname\relax
\typeout{** WARNING: IEEEtran.bst: No hyphenation pattern has been}%
\typeout{** loaded for the language `#1'. Using the pattern for}%
\typeout{** the default language instead.}%
\else
\language=\csname l@#1\endcsname
\fi
#2}}
\providecommand{\BIBdecl}{\relax}
\BIBdecl

\bibitem{Comparative}
S.~Sgouridis, M.~Carbajales-Dale, D.~Csala, M.~Chiesa, and U.~Bardi, ``Comparative net energy analysis of renewable electricity and carbon capture and storage,'' \emph{Nat. Energy}, vol.~4, p. 456–465, 2019.

\bibitem{5618534}
P.~P. Varaiya, F.~F. Wu, and J.~W. Bialek, ``Smart operation of smart grid: Risk-limiting dispatch,'' \emph{Proceedings of the IEEE}, vol.~99, no.~1, pp. 40--57, 2011.

\bibitem{1626410}
O.~Samuelsson, M.~Hemmingsson, A.~Nielsen, K.~Pedersen, and J.~Rasmussen, ``Monitoring of power system events at transmission and distribution level,'' \emph{IEEE Trans. Power Syst.}, vol.~21, no.~2, pp. 1007--1008, 2006.

\bibitem{deka2023learning}
D.~Deka, V.~Kekatos, and G.~Cavraro, ``Learning distribution grid topologies: A tutorial,'' \emph{IEEE Trans. Smart Grid}, vol.~15, no.~1, pp. 999--1013, 2023.

\bibitem{AReview}
F.~Dalavi, M.~E.~H. Golshan, and N.~D. Hatziargyriou, ``A review on topology identification methods and applications in distribution networks,'' \emph{Electric Power Systems Research}, vol. 234, 2024.

\bibitem{7463503}
S.~Chanda and A.~K. Srivastava, ``Defining and enabling resiliency of electric distribution systems with multiple microgrids,'' \emph{IEEE Trans. Smart Grid}, vol.~7, no.~6, pp. 2859--2868, 2016.

\bibitem{7017458}
Z.~Wang and J.~Wang, ``Self-healing resilient distribution systems based on sectionalization into microgrids,'' \emph{IEEE Trans. Power Syst.}, vol.~30, no.~6, pp. 3139--3149, 2015.

\bibitem{Estimating}
D.~Deka, S.~Backhaus, and M.~Chertkov, ``Estimating distribution grid topologies: A graphical learning based approach,'' in \emph{2016 Power Systems Computation Conference (PSCC)}, 2016, pp. 1--7.

\bibitem{Data}
G.~Cavraro, R.~Arghandeh, K.~Poolla, and A.~von Meier, ``{Data-Driven Approach for Distribution Network Topology Detection},'' \emph{IEEE PES General Meeting}, pp. 1--5, 2015.

\bibitem{9094730}
J.~Zhao, L.~Li, Z.~Xu, X.~Wang, H.~Wang, and X.~Shao, ``Full-scale distribution system topology identification using markov random field,'' \emph{IEEE Trans. Smart Grid}, vol.~11, no.~6, pp. 4714--4726, 2020.

\bibitem{8122027}
J.~Yu, Y.~Weng, and R.~Rajagopal, ``Patopa: A data-driven parameter and topology joint estimation framework in distribution grids,'' \emph{IEEE Trans. Power Syst.}, vol.~33, no.~4, pp. 4335--4347, 2018.

\bibitem{On}
O.~Ardakanian, V.~W.~S. Wong, R.~Dobbe, S.~H. Low, A.~von Meier, C.~J. Tomlin, and Y.~Yuan, ``On identification of distribution grids,'' \emph{IEEE Trans. Control Netw. Syst.}, vol.~6, no.~3, pp. 950--960, 2019.

\bibitem{9858017}
Y.~Yuan, S.~H. Low, O.~Ardakanian, and C.~J. Tomlin, ``Inverse power flow problem,'' \emph{IEEE Trans. Control Netw. Syst.}, vol.~10, no.~1, pp. 261--273, 2023.

\bibitem{Structure}
L.~Adair, I.~Shames, and M.~Cantoni, ``{Structure in total least squares parameter estimation for electrical networks},'' \emph{IFAC-PapersOnLine}, vol.~51, no.~23, pp. 420--425, 2018.

\bibitem{9930858}
A.~Mishra and R.~A. de~Callafon, ``Recursive estimation of three phase line admittance in electric power networks,'' \emph{IFAC-PapersOnLine}, vol.~53, no.~2, pp. 34--39, 2020.

\bibitem{Therigidity}
L.~Asimow and B.~Roth, ``{The rigidity of graphs},'' \emph{Trans. Amer. Math. Soc.}, vol. 245, pp. 279--289, 1978.

\bibitem{Basic}
C.~A. Desoer and E.~S. Kuh, \emph{Basic Circuit Theory}.\hskip 1em plus 0.5em minus 0.4em\relax McGraw-Hill, 1969.

\bibitem{Linear}
I.~R. Shafarevich and A.~Remizov, \emph{Linear Algebra and Geometry}.\hskip 1em plus 0.5em minus 0.4em\relax Springer Science \& Business Media, 2012.

\bibitem{Conditions}
B.~Hendrickson, ``{Conditions for unique graph realization},'' \emph{SIAM journal on computing}, vol.~21, pp. 65--84, 1992.

\bibitem{Globally}
D.~Garamvölgyi, S.~J. Gortler, and T.~Jordán, ``Globally rigid graphs are fully reconstructible,'' \emph{Forum of Mathematics, Sigma}, vol.~10, p. e51, 2022.

\bibitem{gortler2014generic}
S.~J. Gortler and D.~P. Thurston, ``Generic global rigidity in complex and pseudo-euclidean spaces,'' \emph{Rigidity and symmetry}, pp. 131--154, 2014.

\bibitem{jordan2017global}
T.~Jord{\'a}n and W.~Whiteley, ``Global rigidity,'' in \emph{Handbook of Discrete and Computational Geometry}.\hskip 1em plus 0.5em minus 0.4em\relax Chapman and Hall/CRC, 2017, pp. 1661--1694.

\bibitem{egres-14-12}
T.~Jord{\'a}n, ``Combinatorial rigidity: graphs and matroids in the theory of rigid frameworks,'' Egerv{\'a}ry Research Group, Budapest, Tech. Rep. TR-2014-12, 2014, {\tt egres.elte.hu}.

\bibitem{oxley2011matroid}
J.~G. Oxley, \emph{Matroid theory}.\hskip 1em plus 0.5em minus 0.4em\relax Oxford University Press, 2011.

\bibitem{rin2024electricgridtopologyadmittance}
\BIBentryALTinterwordspacing
N.~Rin, I.~Shames, I.~R. Petersen, and E.~L. Ratnam, ``Electric grid topology and admittance estimation: Quantifying phasor-based measurement requirements,'' 2024. [Online]. Available: \url{https://arxiv.org/abs/2410.17553}
\BIBentrySTDinterwordspacing

\bibitem{IEEE}
\BIBentryALTinterwordspacing
IEEE. 1991 original test feeders. [Online]. Available: \url{https://cmte.ieee.org/pes-testfeeders/resources/}
\BIBentrySTDinterwordspacing

\end{thebibliography}

\end{document}